\documentclass{article}
\usepackage{amssymb,bm}
\usepackage{amsmath}
\usepackage{amsthm}
\usepackage{amsfonts}
\usepackage{mathrsfs}
\usepackage{appendix}
\usepackage{hyperref}
\usepackage{authblk}
\usepackage{color}

\usepackage{mathtools}

\newtheorem{theorem}{Theorem}[section]

\theoremstyle{remark}

\usepackage{graphics}
\usepackage{graphicx,subcaption}
\usepackage{algorithm}
\usepackage{algorithmic}
\hypersetup{colorlinks,linkcolor={blue},citecolor={blue},urlcolor={red}}

\usepackage{cite}
\usepackage[top=2.5cm, bottom=2.5cm, left=2.5cm, right=2.5cm]{geometry}

\usepackage{threeparttable}
\usepackage{float}
\floatstyle{plaintop}
\restylefloat{table}

\floatstyle{plain}
\restylefloat{figure}
\usepackage{placeins}

\begin{document}

\title{Optimal Pooling Matrix Design for Group Testing with Dilution (Row Degree) Constraints}


\author{Jirong Yi,
        Myung Cho, Xiaodong Wu,
      ~Raghu~Mudumbai, and Weiyu~Xu,
\thanks{Jirong Yi, Raghu Mudumbai, Xiaodong Wu and Weiyu Xu are with Department of Electrical and Computer Engineering, University of Iowa, Iowa City, IA, 52242.}
\thanks{Myung Cho is with Department of Electrical and Computer Engineering, Penn State Behrend, Erie, PA, 16563.}
\thanks{Email: weiyu-xu@uiowa.edu}
}


\maketitle
\begin{abstract}
In this paper, we consider the problem of designing optimal pooling matrix  for group testing (for example, for COVID-19 virus testing)  with the constraint that no more than $r>0$ samples can be pooled together, which we call ``dilution constraint''.  This problem translates to designing a matrix with elements being either 0 or 1 that has no more than $r$  `1's in each row and has a certain performance guarantee of identifying anomalous elements.  We explicitly give pooling matrix designs that satisfy the dilution constraint and have performance guarantees of identifying anomalous elements, and prove their optimality in saving the largest number of tests, namely showing that the designed matrices have the largest width-to-height ratio among all constraint-satisfying 0-1 matrices.

\end{abstract}

\section{Introduction}\label{Sec:Introduction}
In the current COVID-19 pandemic, before effective vaccines are introduced,  large-scale shutdowns can be safely ended if a systematic ``test and trace" program \cite{lee_interrupting_2020,salathe_covid-19_2020} is adopted to contain  the spread of the virus. Widespread availability of mass diagnostic testing is needed for re-opening economy, school and social activities. However, most countries including the US are currently experiencing a scarcity \cite{ranney_critical_2020} of various medical resources including tests \cite{emanuel_fair_2020}, and the test throughput or capacity can be limited. Pooled sample testing has been proposed as a method for increasing the effective capacity of existing testing infrastructure using the classical method of group testing or  newly introduced compressed sensing techniques for virus testing \cite{dorfman_detection_1943, sinnott-armstrong_evaluation_2020,shani-narkiss_efficient_2020,hogan_sample_2020, yelin_evaluation_2020, yi_low-cost_2020,ghosh_compressed_2020} using the RT-qPCR (real-time Quantitative Polymerase Chain Reaction) tests.

However, there is potentially a dilution problem associated with group testing when  samples are pooled together.  When a positive sample is mixed with multiple negative samples,  the concentration of the substances to be tested will be reduced, sometimes below the detection threshold of adopted testing technologies.   This limits the number of samples that can be pooled together.  Recently the FDA of the US has authorized pooled testings for  COVID-19 infection by  Quest Diagnostics, and another pooled testing from the LabCorp.  Possibly partially due to dilution concerns, the FDA's emergency use authorizations (EUA) for pooled testing allows pooling to no more than 4 specimens for the test developed by Quest Diagnostics, and allows pooling to no more than 5 specimens for the test from the LabCorp \cite{office_of_the_commissioner_coronavirus_2020-2,office_of_the_commissioner_coronavirus_2020-1,office_of_the_commissioner_coronavirus_2020,quest_digostics_fda_2020,hale_fda_2020,taylor_high-throughput_2010,taylor_labcorp_2020,quest_diagostics_sars-cov-2_2020}.

Considering the dilution constraint, we ask the following question, ``What is the best pooling design which maximizes test throughput while having a required capability of identifying infected samples, and satisfying the dilution constraint of pooling no more than $r>0$ specimens in each pool?'' In this paper, we give explicit designs for such matrices, and prove their optimality in maximizing test throughput under dilution constraint. We remark that, when there is no dilution constraint, maximizing test throughput under a given infection-identifying capability reduces to a traditional group testing problem.

\section{Problem Formulation and Optimal Designs}
\label{sec:problemformulation}
In this paper, we focus on non-adaptive group testings which deal with $n$ test subjects, and perform $m$ non-adaptive (pooled) tests. Non-adaptive testing has the advantage of reducing the latency of an individual's test because one can infer infection statuses of a group of subjects in a single RT-qPCR run.  We require a pooling matrix design using which we are able to infer the correct infection status for each of these $n$ people, whenever there are no more than $k$ people infected among them. 
It is well-known that the pooling strategy can be represented by an $m \times n$ matrix $A$ which elements being `0' or `1'.  The $j$-th ($1\leq j \leq n$) person's sample is part of the $i$-th pool (test, $1\leq i \leq n$) if and only if  the element of $A$ in its $i$-th row and $j$-th column is `1', namely $A_{i,j}=1$. Moreover, as mentioned above, we require that each row of the pooling matrix have at most $r$ 1's.  Given an identification capability parameter $k$, we would like to design a group testing matrix of dimension $m \times n$ such that $\frac{n}{m}$ is maximized.  In this short paper, we focus on $k=1$, but the results in this paper can be extended to larger $k$ using similar techniques.  

\begin{theorem}
Let us define a  `good' group testing matrix of dimension $m \times n$ as a matrix satisfying the following two conditions:  1) using this matrix we can correctly infer the correct status of each of the $n$ elements whenever there is no more than $k=1$ positive element;  2) and there are no more than $r$ 1's in each row of such a matrix.   Then there is a `good' group testing matrix which satisfies $m=r$, $n=\frac{r(r+1)}{2}$, and  
$$ \frac{n}{m}   =\frac{r+1}{2}.  $$

Moreover, there are no `good' group testing matrices of dimension $m \times n$  for which  

$$ \frac{n}{m}   > \frac{r+1}{2}.  $$

\end{theorem}

\begin{proof}
We first prove that there are no `good' group testing matrices for which $\frac{n}{m}>\frac{r+1}{2}$.  To make sure that we can correctly infer the correct status of each of the $n$ elements whenever there is no more than $k=1$ positive element,  the columns of this matrix must be non-zero columns and must be distinct from each other.   Suppose that we have a `good' group matrices, which has $n_{1}$ columns with a single `1',  $n_{2}$ columns with two `1's, $n_{3}$ columns with three `1's, ...., and $n_{m}$ columns with $m$  `1's. 

Then we have 
$$ \sum_{i=1}^{m} n_{i}=n. $$

Moreover, for every $1\leq i \leq m$, we  have 

$$n_{i} \leq \binom{m}{i}.$$

In addition, the total number of `1's in the matrix is given by  $T=\sum_{i=1}^{m} i \times n_{i}$. Since  the number of `1's in each row is no more than $r$, we have 

$$T\leq r m.$$

We now argue that $n$ cannot be bigger than $m \times \frac{r+1}{2}$. Suppose instead that $n>m \times \frac{r+1}{2}$.  Since there are at most $m$ columns with a single `1', then  there must be at least $$\sum_{i=2}^{m} n_{i}\geq n-m> m \times \frac{r+1}{2}-m=m(\frac{r-1}{2})$$ columns which have at least two `1's.  

We let $l$ be the largest integer such that $$ \sum_{i=1}^{l} \binom{m}{l} \leq n.$$  Then we have 
\begin{align}
&T=\sum_{i=1}^{m} i \times n_{i} \\
&\geq  \sum_{i=1}^{l} i \binom{m}{i} + \left (n-  \sum_{i=1}^{l} \binom{m}{i} \right ) \times (l+1)\\
&> 1 \times m+m(\frac{r-1}{2})\times 2\\
&=mr.
\end{align}
However, this contradicts with $T\leq r m$. So we must have $\frac{n}{m}\leq \frac{r+1}{2}$.

Now we show an explicit construction of a  $m \times n $ `good' group testing matrix with $m=r$ and  $n=\frac{r(r+1)}{2}$.  This matrix has $m=r$ distinct columns with a single `1' (part of this matrix making an $r\times r$ identity matrix), and 
$\binom{m}{2}=\binom{r}{2}=\frac{r(r-1)}{2}$ distinct columns with two  `1's in each column.  So in total there are $r+2 \times \frac{r(r-1)}{2}=r+r(r-1)=r^2$ `1's in this matrix.  By the symmetry of this matrix's construction, each row of this matrix has exactly $r^{2} \times \frac{1}{r}=r$  `1's in each of its rows. 

\end{proof}

\section{Examples for FDA approved Pooling Size $r=4$ or $r=5$}\label{Sec:examples}
From the result in the last section, we have the following optimal design for $r=4$:
\begin{equation}
  A=\begin{bmatrix}
    1 & 0 & 0 & 0 & 1 & 1 & 1 & 0 & 0 & 0\\
    0 & 1 & 0 & 0 & 1 & 0 & 0 & 1 & 1 & 0\\
    0 & 0 & 1 & 0 & 0 & 1 & 0 & 1 & 0 & 1\\
    0 & 0 & 0 & 1 & 0 & 0 & 1 & 0 & 1 & 1\\
  \end{bmatrix},
  \end{equation}

and the following optimal design for $r=5$:

{\center{ \resizebox{0.6\linewidth}{!}{
$ A=\displaystyle \left(  \begin{array}{rrrrrrrrrrrrrrr} 
    1 & 0 & 0 & 0 & 0 & 1 & 1 & 1 & 1 & 0 & 0 & 0 & 0 & 0 & 0\\
    0 & 1 & 0 & 0 & 0 & 1 & 0 & 0 & 0 & 1 & 1 & 1 & 0 & 0 & 0\\
    0 & 0 & 1 & 0 & 0 & 0 & 1 & 0 & 0 & 1 & 0 & 0 & 1 & 1 & 0\\
    0 & 0 & 0 & 1 & 0 & 0 & 0 & 1 & 0 & 0 & 1 & 0 & 1 & 0 & 1\\
    0 & 0 & 0 & 0 & 1 & 0 & 0 & 0 & 1 & 0 & 0 & 1 & 0 & 1 & 1
  \end{array} \right).
$
}
}
}

\section{Conclusions and Future Work}\label{Sec:ConclusionsDiscussions}
It would be interesting to extend this work to adaptive testing, and compressed sensing with dilution constraints \cite{CSwithConstraints}.

\bibliography{01Ref_CS_Virus_Testing}
\bibliographystyle{unsrt}
\end{document}